\documentclass[11pt]{article}

\usepackage{amssymb,amsmath,amsfonts}
\usepackage{graphicx,color,enumitem}
\usepackage{calrsfs} 
\usepackage{amsthm} 
\usepackage{bm}

\usepackage{geometry}

\newcommand{\E}{{\mathbb E}}
\newcommand{\F}{{\mathbb F}}

\renewcommand{\P}{{\mathbb P}}

\newcommand{\R}{{\mathbb R}}

\newcommand{\Fcal}{{\mathcal F}}

\newcommand{\Hcal}{{\mathcal H}}

\newtheorem{theorem}{Theorem}

\newtheorem{lemma}[theorem]{Lemma}

\newtheorem{remark}[theorem]{Remark}

\numberwithin{equation}{section}
\numberwithin{theorem}{section}

\begin{document}

\title{The space of outcomes of semi-static trading strategies\\ need not be closed}
\author{Beatrice Acciaio\thanks{Department of Statistics, London School of Economics and Political Science, 10 Houghton St, WC2A 2AE London, UK, b.acciaio@lse.ac.uk.} \and Martin Larsson\thanks{Department of Mathematics, ETH Zurich, R\"amistrasse 101, CH-8092, Zurich, Switzerland, martin.larsson@math.ethz.ch. Financial support by the Swiss National Science Foundation (SNF) under grant 205121\textunderscore163425 is gratefully acknowledged.} \and Walter Schachermayer\thanks{Fakult\"at f\"ur Mathematik, Universit\"at Wien, Oskar-Morgenstern-Platz 1, A-1090 Wien, walter.schachermayer@univie.ac.at and the Institute for Theoretical Studies, ETH Zurich. Partially supported by the Austrian Science Fund (FWF) under grant P25815, the Vienna Science and Technology Fund (WWTF) under grant MA09-003 and Dr.~Max R\"ossler, the Walter Haefner Foundation and the ETH Zurich Foundation.}}



\maketitle

\begin{abstract}
Semi-static trading strategies make frequent appearances in mathematical finance, where dynamic trading in a liquid asset is combined with static buy-and-hold positions in options on that asset. We show that the space of outcomes of such strategies can have very poor closure properties when all European options for a fixed date $T$ are available for static trading. This causes problems for optimal investment, and stands in sharp contrast to the purely dynamic case classically considered in mathematical finance.
\end{abstract}

\section{Introduction and main results}

Given a local martingale $S$ and a finite stopping time $T$ defined on a stochastic basis $(\Omega,\Fcal,\F,\P)$ in discrete or continuous time, we consider outcomes at time $T$ of semi-static trading in $S$. More specifically, we consider self-financing dynamic trading in $S$ and a risk-free asset with zero interest rate, combined with static (buy-and-hold) positions in arbitrary European options written on the final value $S_T$. Such outcomes are of the form $(H\cdot S)_T+h(S_T)$, where $H\cdot S$ denotes the stochastic integral of the $S$-integrable process $H$ with respect to~$S$, and $h$ is a measurable function satisfying some integrability conditions. The semi-static strategy consists of the pair $(H,h)$, chosen by the investor. This type of semi-static trading strategies has been used extensively in the literature; see e.g.~\cite{Ho11,HK12,BHP13,GHT14,DS14} and the references therein. A key reason is that the collection of time zero prices of all such static claims pins down the law of $S_T$ under $\P$, if $\P$ is the pricing measure.

One could also restrict the static component $h(S_T)$ to lie in a given finite-dimensional set of available options, for instance $h(S_T)=a_0+a_1C_1(S_T)+\cdots+a_nC_n(S_T)$, where $C_i(S_T)=(S_T-K_i)_+$ is a vanilla call payoff with given strike $K_i$, and $a_0,\ldots,a_n\in\R$ are chosen by the investor. Such a setup is also common in the literature; see e.g.~\cite{DH07,ABPS13,BN15}. It is however different from our setting, where $h$ is chosen from an infinite-dimensional space of measurable functions. One of the main purposes of this paper is to clarify the sharply different properties that the two situations may exhibit.

The largest reasonable space of outcomes of semi-static trading strategies is arguably the sum $U+V=\{u+v\colon u\in U, \, v\in V\}$, where
\begin{align*}
U &= \{ (H\cdot S)_T \colon \text{$H$ is $S$-integrable and $H\cdot S$ is a supermartingale on $[0,T]$}\}, \\
V &= L^1(\Omega,\sigma(S_T),\P).
\end{align*}
The supermartingale property of the gains processes $H\cdot S$ is a weak restriction which is implied by any reasonable admissibility or integrability condition that excludes doubling strategies (recall that $S$ from the outset is assumed to be a local martingale). Requiring the static component to be integrable, rather than just measurable, corresponds to a finite initial capital requirement: If an outcome $f=(H\cdot S)_T+h(S_T)$ is integrable---which we interpret as requiring finite initial capital---and if $(H\cdot S)_T\in U$, then $h(S_T)$ is necessarily integrable as well.

On the other hand, the smallest reasonable space of outcomes (at least in our setting without trading constraints) is arguably the sum $U_\infty+V_\infty$, where
\begin{align*}
U_\infty &= \{ (H\cdot S)_T \colon \text{$H$ is $S$-integrable and $H\cdot S$ is a bounded martingale}\}, \\
V_\infty &= L^\infty(\Omega,\sigma(S_T),\P).
\end{align*}
In particular, the dynamic components of such semi-static trading strategies clearly satisfy all admissibility and integrability conditions that have been considered in the literature to date.

The spaces $U$ and $V$ enjoy very strong closure properties. For $V$ this is obvious; for $U$ much less so. Kunita and Watanabe \cite{KW67} proved early on that if $H^n\cdot S$ is a sequence of $\Hcal^2$ martingales such that $(H^n\cdot S)_T\to f$ in $L^2$ for some limit $f$, then the limit is again of the form $f=(H\cdot S)_T$, where $H\cdot S$ is an $\Hcal^2$ martingale; see e.g.~\cite[Theorem~IV.41]{P05}. The same result holds in the $\Hcal^p$ and $L^p$ case for any $p\in(1,\infty]$, and Yor proved that the statement also remains true under uniform (rather than $\Hcal^2$) integrability and $L^1$ (rather than $L^2$) convergence; see \cite{Yo78} and \cite{DS99} for further discussion. In a similar vein, the following result is crucial for the development of arbitrage theory in mathematical finance: if $u_n\in U$, $u_n\ge-1$, and $u_n\to f$ in probability for some random variable $f$, then $f\in U-L^0_+$. That is, $f$ is dominated by some element of~$U$. Further discussion and generalizations can be found e.g.~in \cite{DS94} and \cite{DS98}. Note that these results imply in particular that $U_\infty$ is closed in $L^\infty$, that its closure in any $L^p$ space ($p\ge 1$) is contained in $U$, and that its closure in $L^0$ is contained in $U-L^0_+$.

A natural question is to what extent these closure properties carry over to the spaces $U+V$ and $U_\infty+V_\infty$ of outcomes of semi-static trading strategies. The answer is that they do not. The goal of the present paper is to demonstrate this by way of example. This is done in our two main results, Theorems~\ref{T:Lp} and~\ref{T:Lp_cont}, which cover the discrete and continuous time cases, respectively.

\begin{theorem} \label{T:Lp}
There exists a discrete time stochastic basis $(\Omega,\Fcal,(\Fcal_t)_{t\in\{0,1,2\}},\P)$ with countable sample space $\Omega$, equipped with a bounded martingale $S=(S_t)_{t\in\{0,1,2\}}$ such that the following holds: There exist random variables $g$ and $g_m$, $m\ge1$, such that
\begin{enumerate}
\item\label{T:Lp:1} $g_m\in U_\infty+V_\infty$ and $g_m\ge 0$ for each $m$,
\item\label{T:Lp:2} $g_m \to g$ almost surely and in $L^p$ for every $p\in[1,\infty)$,
\item\label{T:Lp:3} $g \notin U+V-L^0_+$.
\end{enumerate}
\end{theorem}

Thus, the nonnegative random variables $g_m$ are final outcomes of semi-static trading strategies of the most well-behaved kind: their dynamic and static components are both bounded. In particular, the dynamic trading strategies are admissible in the classical sense. Furthermore, the random variables $g_m$ converge to a limit $g$ in a rather strong sense, but this limit cannot be represented as, and not even dominated by, the final outcome of any semi-static trading strategy satisfying minimal regularity conditions. As will become clear from the construction, each $g_m$ can be viewed as a portfolio of digital options, hedged by a position in the underlying stock; see Remark~\ref{R:interpretation} in Section~\ref{S:T:Lp:a}.

To prove Theorem~\ref{T:Lp} we construct final outcomes $g_m$ converging to an integrable limit $g$ which, if it were to have a representation $g\le u+v$ with $u\in U$ and $v\in V$, would violate the simple bound $\|u\|_1+\|v\|_1<\infty$. To achieve this, we construct a sequence of simpler models, each of which admits an element of $U_\infty+V_\infty$ whose $L^p$ norms are small, but whose components in $U_\infty$ and $V_\infty$ are nonetheless large in $L^1$. These models are then pasted together to form a new model, which admits the required sequence of elements $g_m$. The individual models are described in Section~\ref{S:T:Lp:a}, and the pasting procedure is described in Section~\ref{S:T:Lp:b}.

\begin{remark}
Let us mention a conceivable extension of Theorem~\ref{T:Lp}: Is it possible to strengthen part \ref{T:Lp:2} of Theorem~\ref{T:Lp} so that $g_m\to g$ in $L^\infty$? We do not know the answer.
\end{remark}

We emphasize that there is nothing special about discrete time that makes Theorem~\ref{T:Lp} work. An analogous example may be constructed in a basic continuous-time setting, as the following result shows.

\begin{theorem} \label{T:Lp_cont}
There exists a stochastic basis $(\Omega,\Fcal,\F,\P)$ equipped with a Brownian motion $W$ and a stopping time $T$ such that the following holds for the price process $S=W^T$: There exist random variables $g$ and $g_m$, $m\ge1$, such that
\begin{enumerate}
\item\label{T:Lp:1} $g_m\in U_\infty+V_\infty$ and $g_m\ge 0$ for each $m$,
\item\label{T:Lp:2} $g_m \to g$ almost surely and in $L^p$ for every $p\in[1,\infty)$,
\item\label{T:Lp:3} $g \notin U+V-L^0_+$.
\end{enumerate}
Furthermore, $S$ is uniformly bounded.
\end{theorem}

The proof follows the pattern of Theorem~\ref{T:Lp}. The only difference lies in the construction of the individual models, which is presented in Section~\ref{S:T:Lp_cont}. The pasting procedure then works exactly as described in Section~\ref{S:T:Lp:b}, and we refrain from repeating it.

A simple corollary of the above theorems is that the spaces
\[
\left\{(H\cdot S)_T + h(S_T)\colon H\cdot S \text{ is an $\Hcal^p$ martingale}, \, h(S_T)\in L^p(\sigma(S_T))\right\}
\]
need not be closed in $L^p$ ($p\ge1$). The closure of the corresponding space in the case $p=2$ but with finitely many static claims was crucial for the semi-static Jacod-Yor theorem in~\cite{AL:2016}. Thus, we do not expect that result to carry over to the case of infinitely many static claims. The non-closedness of the above spaces is an example of the well-known fact that sums of closed subspaces of Banach spaces need not be closed; see e.g.~Section~41 in~\cite{Halmos:1974}.

Another immediate corollary is that the space
\[
\left\{(H\cdot S)_T + h(S_T) - f\colon \text{$H$ is $1$-admissible}, \, h(S_T)\in L^1(\sigma(S_T)), \, f\in L^0_+\right\}
\]
need not be closed in $L^1$. Here $H$ is called $1$-admissible if it is $S$-integrable and $H\cdot S\ge -1$. This space is a natural space of outcomes in the context of portfolio optimization with semi-static trading opportunities. Thus, existence of optimal strategies is a delicate issue in such a setting.

Finally, we provide a result demonstrating that the non-closedness in Theorems~\ref{T:Lp} and~\ref{T:Lp_cont} is caused by the infinite-dimensionality of the space $V$ of static claims. If $V$ is replaced by a finite-dimensional space, then closedness is retained.

\begin{theorem}
Let $C_1,\ldots,C_n$ be linearly independent elements of $L^1$. The closure in $L^0$ of the space
\[
W=\left\{(H\cdot S)_T + \sum_{i=1}^n a_i C_i \colon \text{$H$ is $1$-admissible}, \, a_1,\ldots,a_n \in \R\right\}
\]
is contained in $W-L^0_+$. Here $H$ is called $1$-admissible if it is $S$-integrable and $H\cdot S\ge -1$.
\end{theorem}

\begin{proof}
Let $\{(H^m\cdot S)_T + h^m:m\ge1\}$ be an $L^0$-convergent sequence in $W$. In particular, it is bounded in $L^0$. By $1$-admissibility, $H^m\cdot S$ is a supermartingale, whence $\E[|(H^m\cdot S)_T|]\le1+\E[1+(H^m\cdot S)_T|]\le 2$, so that the sequence $\{(H^m\cdot S)_T: m\ge1\}$ is bounded in $L^1$ and hence in $L^0$. Thus the sequence $\{h^m:m\ge1\}$ is bounded in~$L^0$. Now write $h^m=r^m \sum_{i=1}^n a^m_i C_i$, where $r^m\ge 0$ and the vector $a^m=(a^m_1,\ldots,a^m_n)$ has unit norm, and take a subsequence to obtain $a^m\to a$ for some unit vector $a$. Thus $\sum_{i=1}^n a^m_i C_i$ converges to a random variable, which is nonzero by linear independence of $C_1,\ldots,C_n$. Boundedness in $L^0$ of $\{h^m:m\ge 1\}$ then implies that $\{r^m:m\ge1\}$ is bounded, hence convergent after passing to a subsequence. To summarize, we have shown that by passing to a subsequence, we may suppose that $h^m$ is convergent in $L^0$. Thus $(H^m\cdot S)_T$ also converges in $L^0$, say to a limit~$f$. By Corollary~4.11 in~\cite{DS99}, this limit is of the form $f=(H\cdot S)_T-g$ for some $1$-admissible $H$ and some $g\in L^0_+$. This proves the result.
\end{proof}

\section{The discrete case} \label{S:T:Lp:a}

The following lemma describes the individual models used in the proof of the discrete time Theorem~\ref{T:Lp}. These individual models are later pasted together according to the procedure described in Section~\ref{S:T:Lp:b}.

\begin{lemma} \label{L:t012 n}
Fix $\varepsilon\in(0,1/2]$, $M>0$, and $a,b\in[2,3]$. There exists a discrete time stochastic basis $(\Omega,\Fcal,(\Fcal_t)_{t\in\{0,1,2\}},\P)$ with finite sample space $\Omega$, equipped with a martingale $S=(S_t)_{t\in\{0,1,2\}}$ with $S_2$ taking values in $\{\pm a,\pm b\}$, as well as a random variable $f$ such that:
\begin{enumerate}
\item\label{L:t012 n:1} $f\in U_\infty+V_\infty$ and $f\ge 0$,
\item\label{L:t012 n:2} $\|f\|_p = M (\varepsilon/2)^{1/p}$ for all $p\in[1,\infty)$,
\item\label{L:t012 n:3} any representation $f\le u+v$ with $u\in U$ and $v\in V$ satisfies $\|u\|_1+\|v\|_1\ge M/16$.
\end{enumerate}
\end{lemma}

\begin{proof}[Proof of Lemma~\ref{L:t012 n}]
The price process $S=(S_t)_{t=0,1,2}$ and filtration $(\Fcal_t)_{t=0,1,2}$ are constructed as follows. Define $S_0=0$ and let $\Fcal_0=\{\emptyset,\Omega\}$. Let $S_1=\pm 1$ with probability $1/2$ each. Next, let $X$ be a Bernoulli random variable with $\P(X=1)=\varepsilon=1-\P(X=0)$, independent of $S_1$. Set $\Fcal_1=\sigma(S_1,X)$. Define the event $A=\{S_1=1\}$, and consider the slightly larger event
\[
\widetilde A = A \cup \{X=1\}.
\]
Now set $S_2=\pm a$ on $\widetilde A$ and $S_2=\pm b$ on $\widetilde A^c$. The martingale condition $\E[S_2\mid\Fcal_1]=S_1$ pins down the conditional probabilities,
\begin{equation} \label{eq:PS2|F1}
\begin{aligned}
\P(S_2=a\mid\Fcal_1) &= \frac{a+S_1}{2a} &&\quad \text{on $\widetilde A$},\\
\P(S_2=b\mid\Fcal_1) &= \frac{b+S_1}{2b} &&\quad \text{on $\widetilde A^c$}.\\
\end{aligned}
\end{equation}
Note that these indeed lie in $(0,1)$ since $a,b\ge 2$ and $S_1=\pm1$. Finally, set $\Fcal=\Fcal_2=\sigma(S_1,X,S_2)$. This completes the description of the stochastic basis $(\Omega,\Fcal,\F,\P)$ and the prices process $S$. In particular, observe that the above construction only involves three independent ``coin flips'' and can thus be accommodated on the eight-point sample space $\Omega=\{0,1\}^3$.

The random variable $f$ is defined to be
\[
f = M (\bm 1_{\widetilde A} - \bm 1_A) = MX\bm 1_{A^c}.
\]
We now prove that $f$ satisfies the properties \ref{L:t012 n:1}--\ref{L:t012 n:3}.

\ref{L:t012 n:1}: Clearly $f\ge0$. Observe that
\begin{equation} \label{eq:digital}
f = -\frac{M}{2} S_1 +  \frac{M}{2} \left( \bm 1_{\widetilde A}-\bm 1_{\widetilde A^c}\right).
\end{equation}
Since $\widetilde A=\{|S_2|=a\}$, it is clear that $f\in U_\infty+V_\infty$.

\ref{L:t012 n:2}: Simply note that $\E[|f|^p] = M^p \varepsilon / 2$.

\ref{L:t012 n:3}: Suppose $f\le u+v$ for some $u\in U$ and $v\in V$. By nonnegativity of $f$ we have $f^2\le fu+fv$. Applying part~\ref{L:t012 n:2} and Lemma~\ref{L:foU+V} below yields
\[
\frac{1}{2}\varepsilon M^2 = \E[f^2] \le \E[fu] + \E[fv] \le 8 \varepsilon M \left( \|u\|_1 + \|v\|_1 \right).
\]
This completes the proof of Lemma~\ref{L:t012 n}.
\end{proof}

The following key property of $f$ was used, which intuitively states that while $f$ is an element of $U+V$, it is almost orthogonal to both $U$ and $V$. This forces the components of $f$ in $U$ and $V$ to be large, despite $f$ itself being rather small. We put ourselves in the setting of the proof of Lemma~\ref{L:t012 n}.

\begin{lemma} \label{L:foU+V}
The random variable $f$ satisfies $\E[fu]\le \varepsilon M \|u\|_1$ for any $u\in U$, and $\E[fv]\le 12 \varepsilon M \|v\|_1$ for any $v\in V$.
\end{lemma}

\begin{proof}
Pick any $u=(H\cdot S)_2 \in U$. In the present discrete setting, $H\cdot S$ is a martingale. Thus, using also the independence of $X$ and $S_1$,
\[
\E[fu]=M\E[X\bm 1_{A^c} (H\cdot S)_1] = M \varepsilon \E[\bm 1_{A^c}(H\cdot S)_1] = M \varepsilon \E[\bm 1_{A^c}\,u] \le M\varepsilon \|u\|_1.
\]
Next, for any $v\in V$,
\[
\E[ fv ] \le M \E[ |v| \E[ X \mid S_2] ].
\]
We claim that $\E[ X \mid S_2]\le 12\varepsilon$, which then completes the proof of the lemma. Since $X=0$ on $\widetilde A^c$, the bound clearly holds on that event. Furthermore, in view of \eqref{eq:PS2|F1} and the fact that $a\in[2,3]$ and $\P(\widetilde A)\ge1/2$, we have $\P(S_2=a)=\E[\bm 1_{\widetilde A}\,\P(S_2=a\mid\Fcal_1)]\ge \frac{1}{2}\times\frac{a-1}{2a}\ge 1/8$. Thus
\[
\E[X\mid S_2=a] \le \frac{\E[X]}{\P(S_2=a)} \le 8\varepsilon,
\]
showing that the claimed bound holds on the event $\{S_2=a\}$. The event $\{S_2=-a\}$ is treated similarly.
\end{proof}

\begin{remark} \label{R:interpretation}
The second part of the representation~\eqref{eq:digital} of the payoff $f$ can be interpreted as a digital option written on the final value $S_2$ of the price process. Indeed, it pays either $+M/2$ if $S_2=\pm a$, or $-M/2$ if $S_2=\pm b$. Thus $f$ can be viewed as a portfolio consisting of a digital option together with the partial hedge $-(M/2)S_1$.
\end{remark}

\section{Pasting together the individual models} \label{S:T:Lp:b}

We now describe the pasting procedure that produces a proof of Theorem~\ref{T:Lp} from the building blocks in Lemma~\ref{L:t012 n}.

Define
\[
\varepsilon_n = 2^{-n^2}, \qquad M_n = 2^n,
\]
and select countably many distinct numbers $a_n$, $b_n$ in the interval $[2,3]$. Now apply Lemma~\ref{L:t012 n} for each~$n$ to obtain stochastic bases $(\Omega_n,\Fcal^n,\F^n, \P_n)$ and corresponding price processes $S^n=(S^n_t)_{t\in\{0,1,2\}}$ and random variables $f_n$ satisfying the properties of Lemma~\ref{L:t012 n}, with $(\varepsilon,M,a,b)$ replaced by $(\varepsilon_n,M_n,a_n,b_n)$.

We now paste these models together. Specifically, define
\[
\Omega = \bigcup_{n\ge1} \Omega_n, \qquad \Fcal_t = \sigma(A: A\in\Fcal^n_t,\, n\ge 1), \qquad \P(\,\cdot\,\mid\Omega_n) = \P_n, \qquad \P(\Omega_n) = 2^{-n},
\]
where $\Omega$ is understood as a disjoint union.  In particular, the collection $\{\Omega_n:n\ge 1\}$ constitutes an $\Fcal_0$-measurable partition of $\Omega$. Next, define the price process by
\[
S_t = \sum_{n\ge 1} S^n_t \bm 1_{\Omega_n},
\]
and let the random variables $g_m$ and $g$ be given by
\[
g_m = \sum_{n=1}^m f_n \bm 1_{\Omega_n}, \qquad g = \sum_{n\ge1} f_n \bm 1_{\Omega_n}.
\]
Clearly $g_m$ converges almost surely to $g$. In fact, the convergence actually takes place in $L^p$ for any $p\in[1,\infty)$. Indeed, writing $\E_n$ for the expectation under $\P_n$, we have by Lemma~\ref{L:t012 n}\ref{L:t012 n:2},
\[
\E[ |g-g_m|^p ]  = \sum_{n=m+1}^\infty 2^{-n}\, \E_n[ |f_n|^p ] =  \frac{1}{2} \sum_{n=m+1}^\infty  2^{-n}\varepsilon_n M_n^p.
\]
Since $2^{-n}\varepsilon_n M_n^p=2^{-n(n+1-p)}$, the right-hand side tends to zero as $m$ tends to infinity.

Moreover, each $g_m$ lies in $U_\infty+V_\infty$. Indeed, for each $n$, Lemma~\ref{L:t012 n}\ref{L:t012 n:3} yields $f_n=(H^n\cdot S^n)_2 + h^n(S^n_2)$ for some $H^n$ and $h^n$ such that the two components are bounded. Thus
\[
f_n\bm 1_{\Omega_n} = (H^n\bm 1_{\Omega_n} \cdot S)_2 + h^n(S_2)\bm 1_{\Omega_n},
\]
and since $\Omega_n=\{|S_2|\in\{a_n,b_n\}\}$, the second term on the right-hand side is a (bounded) function of $S_2$. Thus $f_n\bm 1_{\Omega_n}\in U_\infty+V_\infty$. Since $g_m$ is a finite sum of such terms, it follows that $g_m$ lies in $U_\infty+V_\infty$ as well. Also, $g_m$ is nonnegative since each $f_n$ is nonnegative.

Finally, assume for contradiction that $g$ lies in $U+V-L^0_+$, say $g\le u+v$ with $u=(H\cdot S)_2$ and $v=h(S_2)$. Then $u$ and $v$ lie in $L^1$. On the other hand, by considering the restrictions to $\Omega_n$ we deduce
\[
f_n = g|_{\Omega_n} \le u_n + v_n,
\]
where $u_n=(H|_{\Omega_n}\cdot S^n)_2$ and $v_n=h(S^n_2)$. In view of Lemma~\ref{L:t012 n}\ref{L:t012 n:3}, therefore,
\[
\|u\|_1 + \|v\|_1 = \sum_{n\ge 1} 2^{-n}\,\E_n[ |u_n| + |v_n| ] \ge \frac{1}{16} \sum_{n\ge 1} 2^{-n}M_n = \infty.
\]
This contradiction shows that $g\notin U+V-L^0_+$, and completes the proof of Theorem~\ref{T:Lp}.

\section{The continuous case} \label{S:T:Lp_cont}

The proof of Theorem~\ref{T:Lp_cont} works exactly as the proof of Theorem~\ref{T:Lp}, except that Lemma~\ref{L:t012 n} needs to be replaced by Lemma~\ref{L:t012 n_cont} below when pasting together the individual models.

\begin{lemma} \label{L:t012 n_cont}
Fix $\varepsilon\in(0,1/2]$, $M>0$, and $a,b\in[2,3]$. There exists a stochastic basis $(\Omega,\Fcal,\F,\P)$ equipped with a Brownian motion $W$, a stopping time $T$, and a random variable~$f$ such that the price process $S=W^T$ is bounded with $S_T\in\{\pm a,\pm b\}$, and the random variable~$f$ satisfies
\begin{enumerate}
\item\label{L:t012 n:1} $f\in U_\infty+V_\infty$ and $f\ge 0$,
\item\label{L:t012 n:2} $\|f\|_p = M (\varepsilon/2)^{1/p}$ for all $p\in[1,\infty)$,
\item\label{L:t012 n:3} any representation $f\le u+v$ with $u\in U$ and $v\in V$ satisfies $\|u\|_1+\|v\|_1\ge M/16$.
\end{enumerate}
\end{lemma}

Let $(\Omega,\Fcal,\P)$ be a probability space with a Brownian motion $W$ and an independent Bernoulli random variable $X$ with $\P(X=1)=\varepsilon=1-\P(X=0)$. Let
\[
\sigma = \inf\{t\ge0: |W_t| = 1\}
\]
be the first time the Brownian motion hits level one. Now let the filtration $\F$ be the one generated by the processes $W$ and $X\bm 1_{[\sigma,\infty)}$. Thus, prior to time $\sigma$, only the Brownian motion is observed. Then, at time $\sigma$, the realization $X$ is observed as well. With respect to this filtration, $\sigma$ is a stopping time, $W$ is a Brownian motion, and $X$ is $\Fcal_\sigma$-measurable but independent of $\Fcal_{\sigma-}$.

Next, similarly to the discrete time case, we define the events
\[
A = \{S_\sigma = 1\}, \qquad \widetilde A = A \cup \{X=1\},
\]
and we set
\[
S = W^T,\qquad T = \inf\{t\ge T: |W_t| = a\bm 1_{\widetilde A} + b\bm 1_{\widetilde A^c}\}.
\]
Thus depending on whether $\widetilde A$ or $\widetilde A^c$ occurs, $T$ is the first time the absolute value of the Brownian motion reaches $a$ or $b$, respectively. In particular, $T$ is a stopping time with $T>\sigma$. The price process $S$ is a bounded martingale with $S_T\in\{\pm a,\pm b\}$.

As in the discrete time case, the random variable $f$ is defined to be
\[
f = M (\bm 1_{\widetilde A} - \bm 1_A) = MX\bm 1_{A^c}.
\]
The three properties of Lemma~\ref{L:t012 n_cont} are proved exactly as in the discrete time case, where we use that Lemma~\ref{L:foU+V} remains valid in the present continuous time setting:

\begin{lemma} \label{L:foU+V_cont}
The random variable $f$ satisfies $\E[fu]\le \varepsilon M \|u\|_1$ for any $u\in U$, and $\E[fv]\le 8 \varepsilon M \|v\|_1$ for any $v\in V$.
\end{lemma}

\begin{proof}
Pick any $u=(H\cdot S)_T \in U$, and write $Y = H\cdot S$ for simplicity. Nonnegativity of $f$ and the supermartingale property of $Y$ yield
\[
\E[fu] = \E\left[ f\, \E\left[ Y_T-Y_\sigma\mid\Fcal_\sigma\right]\right] + \E[f\, Y_\sigma] \le \E[f\, Y_\sigma].
\]
Since $X$ is independent of $Y_\sigma\in\Fcal_{\sigma-}$, we have $\E[f\, Y_\sigma]=M\varepsilon\E[\bm 1_{A^c}Y_\sigma]$. The supermartingale property of $Y$ finally yields
\[
\E[\bm 1_{A^c}Y_\sigma] = \E[ Y_\sigma ] + \E[\bm 1_A (Y_T-Y_\sigma)] - \E[\bm 1_A Y_T] \le - \E[\bm 1_A Y_T] \le \| Y_T \|_1 = \|u\|_1,
\]
whence $\E[fu] \le M\varepsilon \|u\|_1$ as claimed.

The statement regarding $v\in V$ follows exactly as in the proof of Lemma~\ref{L:foU+V}, where instead of~\eqref{eq:PS2|F1} one relies on the identities
\[
\begin{aligned}
\P(S_T=a\mid\Fcal_\sigma) &= \frac{a+S_\sigma}{2a} && \quad \text{on $\widetilde A$},\\
\P(S_T=b\mid\Fcal_\sigma) &= \frac{b+S_\sigma}{2b} && \quad \text{on $\widetilde A^c$},\\
\end{aligned}
\]
which are direct consequences of the martingale property of $S$ and the definition of~$T$.
\end{proof}

\bibliographystyle{alpha}
\bibliography{bibl}

\end{document}